%% file: new_CISS_arXiv.tex
\documentclass[conference]{IEEEtran}

\usepackage{amsmath,amsthm}
\usepackage{color}
\usepackage{bbm,bm}
\usepackage{pgfplots}

\usepackage{tikz,tkz-graph}
\usetikzlibrary{matrix,positioning,fit,calc,shapes}
\pgfplotsset{compat=newest}
\usepackage{graphicx}
\usepackage{color}
\usepackage{epsfig, amsmath, amsthm, amssymb,latexsym,amsmath,amsbsy, rotating}
\usepackage{amsfonts}

\newtheorem{theorem}{Theorem}
\newtheorem{lemma}[theorem]{Lemma}

\newtheorem{proposition}[theorem]{Proposition}

\newtheorem{remark}{Remark}

\theoremstyle{definition}

\include{defns}

\newcommand{\ve}{\bf}
\newtheorem{definition}{Definition}

\DeclareFontFamily{U}{futm}{}
\DeclareFontShape{U}{futm}{m}{n}{
  <-> s * [.92] fourier-bb
  }{}
\DeclareSymbolFont{Ufutm}{U}{futm}{m}{n}
\DeclareSymbolFontAlphabet{\mathbb}{Ufutm}

\begin{document}

\title{The Gaussian lossy Gray-Wyner network}

\author{\IEEEauthorblockN{Erixhen Sula and Michael Gastpar}
\IEEEauthorblockA{School of Computer and Communication Sciences, EPFL, CH-1015 Lausanne, Switzerland\\{\texttt{\{erixhen.sula,michael.gastpar\}@epfl.ch}}}}

\maketitle

\begin{abstract}
We consider the problem of source coding subject to a fidelity criterion for the Gray-Wyner network that connects a single source with two receivers via a common channel and two private channels. The pareto-optimal trade-offs between the sum-rate of the private channels and the rate of the common channel is completely characterized for jointly Gaussian sources subject to the mean-squared error criterion, leveraging convex duality and an argument involving the factorization of convex envelopes. Specifically, it is attained by selecting the auxiliary random variable to be jointly Gaussian with the sources.
\end{abstract}
\begin{IEEEkeywords}
Gray-Wyner network, Gaussian optimality, dependent sources.
\end{IEEEkeywords}

\IEEEpeerreviewmaketitle

\section{Introduction}
Source coding for network scenarios has a long history, starting with the work of Slepian and Wolf~\cite{Slepian--Wolf} concerning the distributed compression of correlated sources in a lossless reconstruction setting. In this work, we study a source coding network introduced by Gray and Wyner~\cite{Gray--Wyner}.
In this network, there is a single encoder. It encodes a pair of sources, $(X,Y),$ into three messages, namely, a common message and two private messages.
There are two decoders, both receiving the common message, but each only receiving one of the private messages.
For this problem, both in the setting of lossless and of lossy reconstruction, Gray and Wyner fully characterized the optimal rate(-distortion) regions in~\cite{Gray--Wyner}, up to the optimization over a single auxiliary random variable (which represents the common message).

The contributions of the present paper includes, for the Gaussian lossy Gray-Wyner network under mean-squared error distortion, we prove that it is optimal to select the auxiliary random variable to be jointly Gaussian with the source random variables and for a sufficiently symmetric version, we compute closed-form solutions.

An alternative operational interpretation of the Gray-Wyner network as a model for a caching system has been proposed in~\cite[Section III.C]{Wang--Lim--Gastpar} 


\section{The Gray-Wyner Network}\label{Sec-operational}
Gray and Wyner in~\cite{Gray--Wyner} introduced a particular network source coding problem referred to as the Gray-Wyner network.

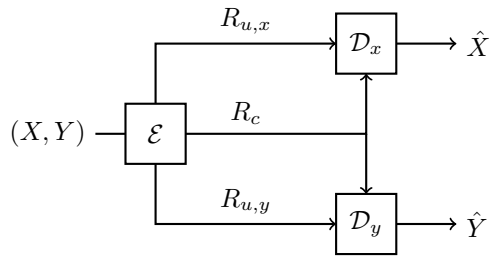
\begin{figure}[ht]%
\centering
\begin{tikzpicture}[scale=0.8]
\draw[black,thick] (0,1) rectangle (1,2);
\draw[black,thick,anchor=east] (-0.5,1.5) node{$(X,Y)$};
\draw[black,thick,anchor=west] (5.5,3) node{$\hat{X}$};
\draw[black,thick,anchor=west] (5.5,0) node{$\hat{Y}$};
\draw[black,thick,anchor=center] (0.5,1.5) node{$\mathcal{E}$};
\draw[black,thick,anchor=south] (2,0) node{$R_{u,y}$};
\draw[black,thick,anchor=south] (2,1.5) node{$R_c$};
\draw[black,thick,anchor=south] (2,3) node{$R_{u,x}$};
\draw[black,thick](-0.5,1.5)--(0,1.5);
\draw[->,black,thick](0.5,1)--(0.5,0)--(3.5,0);
\draw[->,black,thick](1,1.5)--(4,1.5)--(4,2.5);
\draw[->,black,thick](4,1.5)--(4,0.5);
\draw[->,black,thick](0.5,2)--(0.5,3)--(3.5,3);
\draw[black,thick,anchor=center] (4,0) node{$\mathcal{D}_y$};
\draw[black,thick,anchor=center] (4,3) node{$\mathcal{D}_x$};
\draw[black,thick] (3.5,-0.5) rectangle (4.5,0.5);
\draw[black,thick] (3.5,2.5) rectangle (4.5,3.5);
\draw[->,black,thick](4.5,0)--(5.5,0);
\draw[->,black,thick](4.5,3)--(5.5,3);
\end{tikzpicture}
\caption{The Gray-Wyner Network} 
\label{fig:Gray-Wyner}
\end{figure}

The Gray-Wyner network~\cite{Gray--Wyner} is composed of one joint sender and two receivers. The purpose of this network is to convey the joint source $(X,Y)$ (where source $X$ and $Y$ are correlated) to the two receivers, such that each receiver gets only one of the source, either $X$ or $Y$. In other words, receiver or decoder $\mathcal{D}_x$ wants to obtain source $X$, and receiver or decoder $\mathcal{D}_y$ wants to obtain source $Y$. The network is consisting of three links or channels as described in the figure. The central link, of rate $R_{\sc c},$ is provided to both receivers. In addition, each receiver also has access to only one private link. From now on we denote the rates of the private links by $R_{{\sc u}, x}$ and $R_{{\sc u}, y},$ respectively. The main result of~\cite[Theorem 4]{Gray--Wyner}, says that the rate region is given by the closure of the union of the regions
\begin{align} \label{eqn:main}
\mathcal{R}=\{ &(R_{\sc c},R_{{\sc u}, x},R_{{\sc u}, y}):R_{\sc c} \ge I(X,Y;W),  \\
&R_{{\sc u}, x} \ge H(X|W) , R_{{\sc u}, y} \ge H(Y|W) \},\label{eq-GrayWyner}
\end{align}
where the union is over all probability distributions $p(w, x, y)$ with marginals $p(x,y).$

\subsection{Notation}

We use the following notation. Random variables are denoted by uppercase letters and their realizations by lowercase letters. Random column vectors are denoted by boldface uppercase letters and their realizations by boldface lowercase letters. We denote matrices with uppercase letters, e.g., $A,B,C$. For the cross-covariance matrix of $\ve X$ and $\ve Y$, we use the shorthand notation $K_{\ve X \ve Y}$, and for the covariance matrix of a random vector $\ve X$ we use the shorthand notation $K_{\ve X}:= K_{\ve X \ve X}$. In slight abuse of notation, we will let $K_{(X,Y)}$ denote the covariance matrix of the stacked vector $(X,Y)^T.$ We denote the Kullback-Leibler divergence with $D(.||.)$. We denote $\log^{+}{(x)}=\max(\log{x},0)$. 

\section{The Gaussian lossy Gray-Wyner Network}\label{Sec-GaussianGrayWyner}

As in the original work of Gray and Wyner~\cite{Gray--Wyner} (Theorem 8), one may instead ask for {\it lossy}  reconstructions of the original sources $X$ and $Y$ with respect to fidelity criteria. This motivates the following definition (see also the quantity $T(\alpha)$ in~\cite[Remark (4) following Theorem 8]{Gray--Wyner}).


\begin{definition}[Gray-Wyner rate-distortion function]\label{def-Gray-Wyner}
For random variables $X$ and $Y$ with joint distribution $p(x,y),$ the Gray-Wyner rate-distortion function is defined as
\begin{align}
\mathrm{R}_{{\bm{\Delta}}, {\bm{\alpha}}}(X, Y) &= \inf I(X, Y ; W)
\end{align}
such that $I(X; \hat{X}|W)  \leq \alpha_x$ and $I(Y; \hat{Y}|W)\le \alpha_y,$ where the minimum is over all probability distributions $p(\hat{x}, \hat{y}, w, x, y)$ with marginals $p(x,y)$ and satisfying
\begin{align}
{\mathbb E}[d_x(X, \hat{X})] \le \Delta_x & \mbox{ and } {\mathbb E}[d_y(Y, \hat{Y})] \le \Delta_y,
\end{align}
where $d_x(\cdot, \cdot)$ and $d_y(\cdot, \cdot)$ are arbitrary single-letter distortion measures (as in, e.g.,~\cite[Eqn. (30) ff.]{Gray--Wyner}).
\end{definition}

Let us consider a special case of definition \ref{def-Gray-Wyner} for which we can derive a closed-form solution. For a fixed probability distribution $p(x, y),$ we define
\begin{align}
\mathrm{R}_{\Delta, \alpha}(X, Y) &= \inf I(X, Y ; W)
\end{align}
such that $I(X; \hat{X}|W) + I(Y; \hat{Y}|W)\le \alpha,$
where the minimum is over all probability distributions $p(\hat{x}, \hat{y}, w, x, y)$ with marginals $p(x,y)$ and satisfying
\begin{align}
{\mathbb E}[d_x(X, \hat{X})] \le \Delta & \mbox{ and } {\mathbb E}[d_y(Y, \hat{Y})] \le \Delta,
\end{align}
where $d_x(\cdot, \cdot)$ and $d_y(\cdot, \cdot)$ are arbitrary single-letter distortion measures (and $\Delta_x=\Delta_y=\Delta$).
Another equivalent way of writing would be
\begin{align}
\mathrm{R}_{\Delta, \alpha}(X, Y) &= \min_{\alpha_x+\alpha_y = \alpha} \mathrm{R}_{ {\bm{\Delta}}, {\bm{\alpha}}}(X, Y).
\end{align}

%

\begin{theorem}\label{thm-Gauss-lossymum}
Let $X$ and $Y$ be jointly Gaussian with mean zero, equal variance $\sigma^2,$ and with correlation coefficient $\rho.$
Let $d_x(\cdot, \cdot)$ and $d_y(\cdot, \cdot)$ be the mean-squared error distortion measure. Then,
\begin{align} \label{eqn:piecedouble}
\lefteqn{ \mathrm{R}_{\Delta, \alpha}(X, Y) } \nonumber \\
 &= \left\{ \begin{array}{lr} \frac{1}{2} \log^{+}{\frac{1+\rho}{2\frac{\Delta }{\sigma^2}e^{\alpha}+\rho-1}}, &   \mbox{ if } \sigma^2(1-\rho) \le \Delta e^{\alpha} \le \sigma^2  \\
 \frac{1}{2} \log^{+}{\frac{1-\rho^2}{\frac{\Delta^2}{\sigma^4}e^{2\alpha}}} , &   \mbox{ if }  \Delta e^\alpha \le \sigma^2(1-\rho).
\end{array} \right. 
\end{align}
\end{theorem}
The proof of this theorem is given in Section~\ref{sec-proof-thm-Gauss-lossymum}.
\begin{remark}
Assuming that auxiliaries are jointly Gaussian with the sources, the same formula was derived in~\cite[Theorem 4.3]{EPFL8001} via a different reasoning.
\end{remark}
Figure \ref{fig:plot} will illustrate the piecewise function of (\ref{eqn:piecedouble}) in terms of $\Delta e^{\alpha}$, for the specific choice of $\rho=0.5$ and $\sigma^2=1$.
\begin{figure}[ht]%
\centering
\scalebox{0.8}{\input{gray_wyner.tex}}
\caption{Piecewise function, $\mathrm{R}_{\Delta, \alpha}(X, Y)$ versus $\Delta e^{\alpha}.$}
\label{fig:plot}
\end{figure}
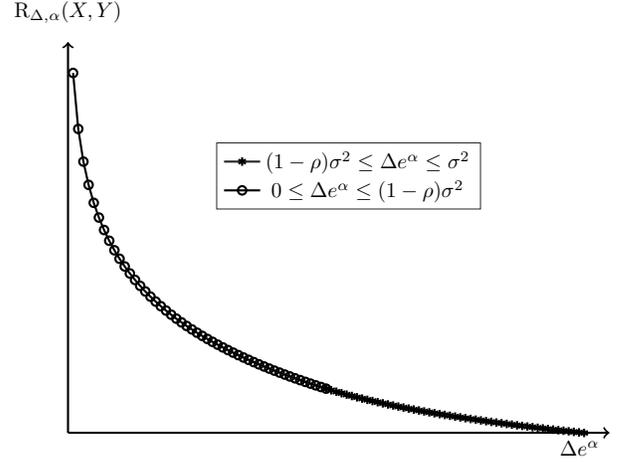

\section{Proof of Theorem~\ref{thm-Gauss-lossymum}}\label{sec-proof-thm-Gauss-lossymum}

\subsection{Preliminary Results for the Proof of Theorem~\ref{thm-Gauss-lossymum}}

The following results are used as intermediate tools in the proof of the main results.
\begin{theorem} \label{Thm:Hypercontract}
For $K \succeq 0$, $0 < \lambda < 1$, there exists a $0\preceq K^{\prime} \preceq K$ and $(X^{\prime},Y^{\prime})\sim \mathcal{N}(0,K^{\prime})$ such that $(X, Y)\sim p_{X,Y}$ with covariance matrix $K$ the following inequality holds:
\begin{align}
\inf_W h(Y|W)+& h(X|W)  - (1+\lambda) h(X, Y|W)  \\
 &\geq h(Y^{\prime})+ h(X^{\prime}) -(1+\lambda)h(X^{\prime}, Y^{\prime}) . \label{Eq-Thm:Hypercontract}
\end{align}
\end{theorem}
\begin{proof}
The theorem is a consequence of \cite[Theorem~2]{Hyper_Gauss}, for a specific choice of $p=\frac{1}{\lambda}+1$. The rest of the proof is given in Appendix \ref{app:HyperGauss}.
\end{proof}

Let $K$ be the covariance matrix with unity entries in the main diagonal and $\rho$ entries in the off-diagonal.

\begin{lemma} \label{lem:lemmasymmery}
For $(X^{\prime},Y^{\prime})\sim \mathcal{N}(0,K^{\prime})$, the following inequality holds
\begin{align}
\min_{K^{\prime}: 0 \preceq K^{\prime} \preceq \begin{pmatrix} 1 & \rho \\ \rho &1  \end{pmatrix}} & h(X^{\prime})+h(Y^{\prime})-(1+\lambda)h(X^{\prime},Y^{\prime}) \\
& \hspace{-5em} \geq  \frac{1}{2} \log{\frac{1}{1-\lambda^2}}-\frac{\lambda}{2} \log{(2\pi e)^2\frac{(1-\rho)^2(1+\lambda)}{1-\lambda}},
\end{align}
where $\lambda \leq \rho$.
\end{lemma}
\begin{proof}
The proof is given in Appendix \ref{App:lowerboundmainRWCI}.
\end{proof}

\subsection{Lower Bound}
First, we observe that for mean-squared error, the source variance is irrelevant: A scheme attaining distortion $\Delta$ for sources of variance $\sigma^2$
is a scheme attaining distortion $\Delta/\sigma^2$ on unit-variance sources, and vice versa.
Therefore, for ease of notation, in the sequel, we assume that the sources are of unit variance.
Then, we can bound:
\begin{align}
&\mathrm{R}_{\Delta, \alpha}(X, Y)\\
&= \inf_{ \substack{W,\hat{X},\hat{Y}:I(X;\hat{X}|W)+I(Y;\hat{Y}|W) \leq \alpha\\ {\mathbb E}[(X-\hat{X})^2] \leq \Delta \\ {\mathbb E}[(Y-\hat{Y})^2] \leq \Delta }} I(X,Y;W)  \\
& \geq \inf_{ \substack{W,\hat{X},\hat{Y}: {\mathbb E}[(X-\hat{X})^2] \leq \Delta \\ {\mathbb E}[(Y-\hat{Y})^2] \leq \Delta }} I(X,Y;W) + \nu I(X;\hat{X}|W) \nonumber  \\
& \quad \quad +\nu(I(Y;\hat{Y}|W) - \alpha) \label{eqn:weakduality} \\
& = \inf_{ \substack{W,\hat{X},\hat{Y}: {\mathbb E}[(X-\hat{X})^2] \leq \Delta \\ {\mathbb E}[(Y-\hat{Y})^2] \leq \Delta }} \hspace{-1em} h(X,Y) - \nu \alpha +\nu (h(X|W)+h(Y|W)) \nonumber \\
& \quad \quad -h(X,Y|W)  -\nu(h(X|W,\hat{X})+h(Y|W,\hat{Y})) \\
& \geq   h(X,Y) - \nu \alpha + \nu \inf_W h(X|W)+h(Y|W)- \frac{1}{\nu}h(X,Y|W) \nonumber \\
& \quad \quad  + \inf_{ \substack{W,\hat{X},\hat{Y}: {\mathbb E}[(X-\hat{X})^2] \leq \Delta \\ {\mathbb E}[(Y-\hat{Y})^2] \leq \Delta }} -\nu(h(X|W,\hat{X})+h(Y|W,\hat{Y})) \label{eqn:minsplit} \\
& \geq h(X,Y) - \nu \alpha + \nu \cdot \hspace{-2em} \min_{0 \preceq K^{\prime} \preceq \begin{pmatrix} 1 & \rho \\ \rho &1  \end{pmatrix}} \hspace{-2em} h(X^{\prime})+h(Y^{\prime})- \frac{1}{\nu}h(X^{\prime},Y^{\prime})  \nonumber \\
& +\nu \cdot \left( \min_{ \substack{(W,\hat{X},\hat{Y}) \in \mathcal{P}_G: \\ {\mathbb E}[(X-\hat{X})^2] \leq \Delta }} \hspace{-1.5em} -h(X|W,\hat{X})+  \min_{ \substack{(W,\hat{X},\hat{Y}) \in \mathcal{P}_G: \\ {\mathbb E}[(Y-\hat{Y})^2] \leq \Delta }} \hspace{-1.5em} -h(Y|W,\hat{Y}) \right)   \label{eqn:twobound} \\
& = h(X,Y) - \nu \alpha -\nu \log{(2 \pi e \Delta)}  \nonumber \\
& \quad \quad + \nu \cdot \hspace{-1.5em} \min_{0 \preceq K^{\prime} \preceq \begin{pmatrix} 1 & \rho \\ \rho &1  \end{pmatrix}} h(X^{\prime})+h(Y^{\prime})- \frac{1}{\nu}h(X^{\prime},Y^{\prime}) \label{eqn:twoboundeval} \\
& = \frac{1}{2} \log{(2\pi e)^2(1-\rho^2)} - \nu \alpha -\nu \log{(2 \pi e \Delta)} \nonumber \\
& \quad \quad  + \frac{\nu}{2} \log{\frac{\nu^2}{2\nu-1}}-\frac{1-\nu}{2} \log{(2\pi e)^2\frac{(1-\rho)^2}{2\nu-1}} \label{eqn:hypereval2} \\
& = \left\{ \begin{array}{lr} \frac{1}{2} \log^{+}{\frac{1+\rho}{2\Delta e^{\alpha}+\rho-1}}, &   \mbox{ if } 1-\rho \le \Delta e^{\alpha} \le 1  \\
 \frac{1}{2} \log^{+}{\frac{1-\rho^2}{\Delta^2e^{2\alpha}}} , &   \mbox{ if }  \Delta e^\alpha \le 1-\rho.
\end{array} \right. \label{eqn:lastGrayWyner}
\end{align}
where~\eqref{eqn:weakduality} follows from weak duality for $ \nu \geq 0;$ \eqref{eqn:minsplit} follows from bounding the infimum of the sum with the sum of the infima of its summands, and the fact that relaxing the constraints cannot increase the value of the infimum;
\eqref{eqn:twobound} follows from Theorem~\ref{Thm:Hypercontract} where $\nu:=\frac{1}{1+\lambda}$ and for the constraint $0 \leq \lambda<1$ (indeed we can also include zero) to be satisfied we need $ \frac{1}{2} <\nu \leq1$ and \cite[Lemma~1]{Thomas} on each of the terms;
\eqref{eqn:twoboundeval} follows by observing
\begin{align}
h(X|W,\hat{X}) & = h(X-\hat{X}|W,\hat{X}) \\
  & \le h(X-\hat{X}) \\
  & \le \frac{1}{2} \log (2 \pi e \Delta),
\end{align}
where the last step is due to the fact that ${\mathbb E}[(X-\hat{X})^2] \leq \Delta$; (\ref{eqn:hypereval2}) follows from Lemma \ref{lem:lemmasymmery} for $\nu \geq \frac{1}{1+\rho}$; and (\ref{eqn:lastGrayWyner}) follows from maximizing 
\begin{align}
&\ell(\nu):=\frac{1}{2} \log{(2\pi e)^2(1-\rho^2)} - \nu \alpha -\nu \log{(2 \pi e \Delta)} \\
&\quad \quad \quad  + \frac{\nu}{2} \log{\frac{\nu^2}{2\nu-1}}-\frac{1-\nu}{2} \log{(2\pi e)^2\frac{(1-\rho)^2}{2\nu-1}},
\end{align}
for $1 \geq \nu \geq \frac{1}{1+\rho}$. Now we need to choose the tightest bound $\max_{1 \geq \nu \geq \frac{1}{1+\rho}} \ell(\nu)$. Note that the function $\ell$ is concave since
\begin{align}
\frac{\partial^2 \ell}{\partial \nu^2}=-\frac{1}{\nu(2\nu-1)}<0.
\end{align}
Since it also satisfies monotonicity
\begin{align}
\frac{\partial \ell}{\partial \nu}=\log{\frac{\nu(1-\rho)}{(2\nu-1)\Delta e^{\alpha}}},
\end{align}
its maximal value occurs when the derivative vanishes, that is, when $\nu_*=\frac{\Delta e^{\alpha}}{2\Delta e^{\alpha}-1+\rho}.$ Substituting for the optimal $\nu_*$ we get
\begin{align}
\mathrm{R}_{\Delta, \alpha}(X, Y) \geq \ell \left(\frac{\Delta e^{\alpha}}{2\Delta e^{\alpha}-1+\rho} \right) =\frac{1}{2} \log^+ \frac{1+\rho}{2 \Delta e^{\alpha}-1+\rho},
\end{align}
for $1 \geq \nu_* \geq \frac{1}{1+\rho}$, which means the expression is valid for $1-\rho \leq \Delta e^{\alpha} \leq 1$. 

The other case is $ \Delta e^{\alpha} \leq 1-\rho$. In this case note that $\nu(1-\rho) \geq \nu \Delta e^{\alpha} \geq (2\nu-1)\Delta e^\alpha$ for $\nu \leq 1$. This implies $\frac{\nu(1-\rho)}{(2\nu-1)\Delta e^{\alpha}} \geq 1$, thus we have $\frac{\partial \ell}{\partial \nu} \geq 0$. Since the function is concave and increasing the maximum is attained at $\nu_*=1$, thus 
\begin{align}
\mathrm{R}_{\Delta , \alpha}(X, Y) \geq \ell \left(1 \right) =\frac{1}{2} \log^+ \frac{1-\rho^2}{\Delta^2e^{2\alpha}},
\end{align}
where the expression is valid for $\Delta e^{\alpha} \leq 1-\rho$.
As stated at the beginning of the proof, this is the correct formula assuming unit-variance sources.
For sources of variance $\sigma^2,$ it suffices to replace $\Delta$ with $\Delta/\sigma^2,$ which leads to the expression given in the theorem statement.

\subsection{Upper Bound}
The upper bound follows by plugging in jointly Gaussian random variables, that was derived in in~\cite[Theorem 4.3]{EPFL8001}.

\section{Conclusion and Discussion}
For the Gaussian lossy Gray-Wyner network under mean-squared error distortion, the rate region $(R_{u,x}+R_{u,y},R_c)$, which is the sum-rate of the private channels and the rate of the common channel is fully characterized. 

\appendices 

\section{Proof of Theorem \ref{Thm:Hypercontract}} \label{app:HyperGauss}
The techniques to establish the optimality of Gaussian distributions is used in \cite{Geng--Nair} and is known as factorization of lower convex envelope. Let us define the following object
\begin{align} \label{eqn:convexenv_object}
V(K)&=\inf_{(X,Y):K_{(X,Y)}=K} \inf_W h(Y|W)+h(X|W) \\
& \quad \quad - (1+\lambda) h(X, Y|W),
\end{align}
where $\lambda$ is a real number, $0<\lambda<1$  and $K$ is an arbitrary covariance matrix. Let $\ell_{\lambda}(X,Y)=h(Y)+h(X) - (1+\lambda) h(X,Y)$, and $\breve{\ell}_{\lambda}(X,Y) =\inf_W h(Y|W)+h(X|W) - (1+\lambda) h(X, Y|W)$, where $\breve{\ell}_{\lambda}(X,Y)$ is the lower convex envelope of $\ell_{\lambda}(X,Y)$. 

First, we prove that the infimum is attained, then we prove that a Gaussian $W$ attains the infimum in Equation~\eqref{eqn:convexenv_object}. Together, these two arguments establish Theorem \ref{Thm:Hypercontract}.

\subsection{The infimum in Equation~\eqref{eqn:convexenv_object} is attained}\label{app:HyperGauss:infattained}
\begin{proposition}[Proposition 17 in \cite{Geng--Nair}]
Consider a sequence of random variables $\left\{X_n,Y_n\right\}$ such that $K_{(X_n,Y_n)}\preceq K$ for all $n$, then the sequence is tight.
\end{proposition}

\begin{theorem}[Prokhorov] \label{thm:weakconvergence}
If $\left\{X_n,Y_n\right\}$ is a tight sequence then there exists a subsequence $\left\{X_{n_i},Y_{n_i}\right\}$ and a limiting probability distribution $\left\{X_*,Y_*\right\}$ such that $\left\{X_{n_i},Y_{n_i}\right\} \overset{w}{\Rightarrow} \left\{X_*,Y_*\right\}$ converges weakly in distribution.
\end{theorem} 

Note that $\ell_{\lambda}(X,Y)=h(Y)+h(X) - (1+\lambda) h(X,Y)$ can be written as $(1+\lambda)I(X;Y)-\lambda[h(X)+h(Y)]$. Thus, it is enough to show that this expression is lower semi-continuous. We will show by utilizing the following theorem.

\begin{theorem}[\cite{Posner}]
If $p_{X_n,Y_n} \overset{w}{\Rightarrow} p_{X,Y}$ and $q_{X_n,Y_n} \overset{w}{\Rightarrow} q_{X,Y}$, then $D(p_{X,Y}||q_{X,Y}) \leq \liminf\limits_{n \to \infty} D(p_{X_n,Y_n}||q_{X_n,Y_n})$.
\end{theorem}


Observe that $I(X;Y)=D(p_{X,Y}||q_{X,Y})$, where $q_{X,Y}=p_Xp_Y$. For the theorem to hold we need to check the assumptions.
First, from Theorem~\ref{thm:weakconvergence}, we have $p_{X_n,Y_n} \overset{w}{\Rightarrow} p_{X,Y}.$
Second, since the marginal distributions converge weakly if the joint distribution converges weakly, we also have $q_{X_n,Y_n} \overset{w}{\Rightarrow} q_{X,Y}.$
Therefore,
\begin{align} \label{eqn:semicont1}
I(X;Y) \leq \liminf\limits_{n \to \infty} I(X_n;Y_n).
\end{align}
To preserve the covariance matrix $K_{(X,Y)}$, there are three degrees of freedom plus one degree of freedom coming from minimizing the objective, thus $|\mathcal{W}| \leq 4$ is enough to attain the minimum.  

Let us introduce $\delta > 0$ and define $N_{\delta}\sim \mathcal{N}(0,\delta)$, being independent of $\{X_n\},X,\{Y_n\}$ and $Y$. From the entropy power inequality, we have 
\begin{align} \label{eqn:semicont2}
h(X_n+N_{\delta}) &\geq h(X_n) \\
h(Y_n+N_{\delta}) &\geq h(Y_n),
\end{align}
and moreover, for Gaussian perturbations, we have
\begin{align} \label{eqn:semicont3}
\liminf_{n\to \infty} h(X_n+N_{\delta}) =h(X+N_{\delta}).
\end{align}
This results in 
\begin{align}
& \liminf_{n \to \infty} \ell_{\lambda}(X_n,Y_n) \nonumber \\
&=\liminf_{n \to \infty} (1+\lambda)I(X_n;Y_n)-\lambda[h(X_n)+h(Y_n)] \\
& \geq \liminf_{n \to \infty} (1+\lambda)I(X_n;Y_n)-\lambda[h(X_n+N_{\delta})+h(Y_n+N_{\delta})] \label{eqn:EPIbound} \\
& \geq (1+\lambda)I(X;Y)-\lambda[h(X+N_{\delta})+h(Y+N_{\delta})], \label{eqn:semicont4}
\end{align}
where (\ref{eqn:EPIbound}) follows from (\ref{eqn:semicont2}) and (\ref{eqn:semicont4}) follows from (\ref{eqn:semicont1}), (\ref{eqn:semicont3}). Letting $\delta \to 0$, we obtain the weak semicontinuity of our object $\liminf_{n \to \infty} \ell_{\lambda}(X_n,Y_n) \geq  \ell_{\lambda}(X,Y)$.

\subsection{A Gaussian auxiliary $W$ attains the infimum in Equation~\eqref{eqn:convexenv_object}}\label{app:HyperGauss:Gauss}
This proof follows from~\cite{Hyper_Gauss}.

\section{Proof of Lemma \ref{lem:lemmasymmery}} \label{App:lowerboundmainRWCI}
Let $K^{\prime}$ be parametrized as $K^{\prime} = \begin{pmatrix} \sigma^2_X & q\sigma_X \sigma_Y \\ q\sigma_X \sigma_Y & \sigma^2_Y  \end{pmatrix} \succeq 0$. Then, the problem is the same to the following one

\begin{align}
&\min_{K^{\prime}: 0 \preceq K^{\prime} \preceq \begin{pmatrix} 1 & \rho \\ \rho &1  \end{pmatrix}} h(X^{\prime})+h(Y^{\prime})-(1+\lambda)h(X^{\prime},Y^{\prime}) \nonumber \\
& \quad \quad = \min_{(\sigma_X,\sigma_Y,q) \in \mathcal{A}_{\rho}} \frac{1}{2}\log{(2 \pi e)^2 \sigma_X^2\sigma_Y^2} \\
& \hspace{5em} -\frac{1+\lambda}{2}\log{(2 \pi e)^2 \sigma_X^2\sigma_Y^2(1-q^2)} \label{eqn:mlproof1}
\end{align}
where the set
\begin{align} 
\mathcal{A}_{\rho}=\left\{( \sigma_X,\sigma_Y,q): \begin{pmatrix} \sigma^2_X -1& q\sigma_X \sigma_Y -\rho\\ q\sigma_X \sigma_Y-\rho & \sigma^2_Y -1 \end{pmatrix} \preceq 0 \right\}.
\end{align}
Matrices of dimension $2\times2$ are negative semi-definite if and only if the trace is negative and determinant is positive. Thus, we can rewrite the set as 
\begin{align} 
\mathcal{A}_{\rho}=\left\{(\sigma_X,\sigma_Y,q): \hspace{-1.6em} \substack{\sigma^2_X+\sigma^2_Y \leq 2, \\ \quad (1-q^2)\sigma^2_X\sigma^2_Y +2\rho q \sigma_X\sigma_Y +1-\rho^2-(\sigma^2_X+\sigma^2_Y) \geq 0} \right\}.
\end{align}
By making use of $\sigma^2_X+\sigma^2_Y \geq 2\sigma_X\sigma_Y$, we derive that $\mathcal{A}_{\rho} \subseteq \mathcal{B}_{\rho}$, where
\begin{align} 
\mathcal{B}_{\rho}=\left\{ (\sigma_X,\sigma_Y,q): \hspace{-1.2em} \substack{\sigma_X\sigma_Y \leq 1, \\ \quad (1-q^2)\sigma^2_X\sigma^2_Y +2\rho q \sigma_X\sigma_Y +1-\rho^2-2\sigma_X\sigma_Y) \geq 0} \right\}.
\end{align} We will further reparametrize and define $\sigma^2=\sigma_X\sigma_Y$, thus 
\begin{align} 
\mathcal{D}_{\rho}=\left\{( \sigma^2,q): \substack{\sigma^2 \leq 1, \\ (\sigma^2(1-q)-1+\rho)(\sigma^2(1+q)-1-\rho) \geq 0} \right\}.
\end{align}
The second equation in the definition of the set $\mathcal{D}_{\rho}$ has roots $\sigma^2=\frac{1+\rho}{1+q}$ and $\sigma^2=\frac{1-\rho}{1-q}$, thus the inequality is true if $\sigma^2$ is not in between these two roots. Thus, we can rewrite the set $\mathcal{D}_{\rho}$ as
\begin{align} 
\mathcal{D}_{\rho}=\left\{( \sigma^2,q): \substack{\rho \geq q, \quad \sigma^2(1-q) \leq 1-\rho \\ \rho < q, \quad  \sigma^2(1+q) \leq 1+\rho } \label{eqn:D} \right\}.
\end{align}
Thus, we have
\begin{align}
& \min_{(\sigma_X,\sigma_Y,q) \in \mathcal{A}_{\rho}} \frac{1}{2}\log{(2 \pi e)^2 \sigma_X^2\sigma_Y^2} \\
& \quad -\frac{1+\lambda}{2}\log{(2 \pi e)^2 \sigma_X^2\sigma_Y^2(1-q^2)} \geq \min_{(\sigma^2,q) \in \mathcal{D}_{\rho}} f(\lambda,\sigma^2,q)
\label{eqn:mlproof2}
\end{align}
where,
\begin{align}f(\lambda,\sigma^2,q)&=\frac{1}{2}\log{(2 \pi e)^2 \sigma^4}-\frac{1+\lambda}{2}\log{(2 \pi e)^2 \sigma^4(1-q^2)} \label{eqn:f}.
\end{align}
For now let us assume $\rho$ is positive and start from the case $\rho \geq q$. Then, by weak duality we have 
\begin{align}
\min\limits_{(\sigma^2,q) \in \mathcal{D}_{\rho}} \hspace{-1em} f(\lambda,\sigma^2,q) \geq \min\limits_{\sigma^2,q} f(\lambda,\sigma^2,q) + \mu(\sigma^2(1-q)-1+\rho)), \label{eqn:mlproof3}
\end{align}
for any $\mu \geq 0$.
By applying Karush-Kuhn-Tucker (KKT) conditions we get
\begin{align}
\frac{\partial }{\partial \sigma^2}=-\frac{\lambda}{\sigma^2} + \mu (1-q)&=0, \label{eqn:KKT1} \\ 
\frac{\partial }{\partial q}=\frac{(1+\lambda)q}{1-q^2} - \mu \sigma^2&=0, \label{eqn:KKT2} \\
\mu(\sigma^2(1-q)-1+\rho))&=0, \label{eqn:KKT3}
\end{align}
where (\ref{eqn:KKT1}), (\ref{eqn:KKT2}) is known as stationary condition and (\ref{eqn:KKT3}) is known as complementary slackness condition. By using (\ref{eqn:KKT1}) we get 
\begin{align}
\mu=\frac{\lambda}{\sigma^2(1-q)} \label{eqn:KKT4}.
\end{align}
By using (\ref{eqn:KKT2}) we get 
\begin{align}
\mu=\frac{(1+\lambda)q}{\sigma^2(1-q^2)} \label{eqn:KKT5}.
\end{align}
By equating (\ref{eqn:KKT4}) and (\ref{eqn:KKT5}) we deduce that $q_*=\lambda$. Since $\lambda>0$, then $\mu \neq 0$ and by using (\ref{eqn:KKT3}) we get $\sigma^2_*=\frac{1-\rho}{1-\lambda}$. In addition, $\mu_*=\frac{\lambda}{1-\rho}$. Since the KKT conditions are satisfied by $q_*,\sigma^2_*$ and $\mu_*$ then strong duality holds, thus 
\begin{align}
& \min\limits_{\sigma^2,q} f(\lambda,\sigma^2,q) + \mu(\sigma^2(1-q)-1+\rho))= f(\lambda,\frac{1-\rho}{1-\lambda},\lambda) \nonumber \\
&= \frac{1}{2} \log{\frac{1}{1-\lambda^2}}-\frac{\lambda}{2} \log{(2\pi e)^2\frac{(1-\rho)^2(1+\lambda)}{1-\lambda}}. \label{eqn:mlprooffinal}
\end{align}
By combining (\ref{eqn:mlproof1}), (\ref{eqn:mlproof2}), (\ref{eqn:mlproof3}) and (\ref{eqn:mlprooffinal}) we get the desired lower bound. 

For the case $\rho < q$, let us optimize over $\sigma^2$ for any fixed $q$. The function $f$ is decreasing in $\sigma^2$. Also, the function $f$ is convex in $\sigma^2$. Since the object is continuous in $\sigma^2$ and the constraint is linear for any fixed $q$, then the optimal choice is $\sigma^2=\frac{1+\rho}{1+q}$. Thus, 
\begin{align}
\min\limits_{(\sigma^2,q) \in \mathcal{D}_{\rho}} f(\lambda,\sigma^2,q) \geq \min\limits_{q \in [\rho, 1]} f(\lambda,\frac{1+\rho}{1+q},q). \label{eqn:mlproof4}
\end{align}
The function on the right hand side can be written as
\begin{align}
f(\lambda,\frac{1+\rho}{1+q},q)&= \frac{1}{2}\log{(2\pi e)^2\frac{(1+\rho)^2}{(1+q)^2}} \\
& \quad -\frac{1+\lambda}{2}\log{(2\pi e)^2\frac{(1+\rho)^2(1-q)}{(1+q)}}.
\end{align}
The function is convex and increasing in $q$ for $q \in [\rho, 1]$, 
\begin{align}
\frac{\partial f}{\partial q}&=\frac{q+\lambda}{1-q^2}>0, \\
\frac{\partial^2 f}{\partial q^2}&=\frac{1+q^2+2\lambda q}{(1-q^2)^2} >0,
\end{align}
thus, the optimal value of $q=\rho$, is guaranteed to give the minimum. To conclude we show that $f(\lambda,\frac{1-\rho}{1-\lambda},\lambda) \leq f(\lambda,1,\rho)$ for $\lambda \leq \rho$. To show this we define
\begin{align}
h(\lambda)&:=f(\lambda,\frac{1-\rho}{1-\lambda},\lambda)-f(\lambda,1,\rho)\\
&=\frac{1}{2}\log{\frac{1-\rho^2}{1-\lambda^2}}- \frac{\lambda}{2}\log{\frac{(1+\lambda)(1-\rho)}{(1-\lambda)(1+\rho)}},
\end{align}
and the new defined function is increasing in $\lambda$,
\begin{align}
\frac{\partial h}{\partial \lambda}&=-\frac{1}{2} \log{\frac{(1+\lambda)(1-\rho)}{(1-\lambda)(1+\rho)}} \geq 0, \quad \text{for } \lambda\leq \rho 
\end{align}
and it is concave in $\lambda$,
\begin{align}
\frac{\partial^2 h}{\partial \lambda^2}&=-\frac{1}{1-\lambda^2}<0,
\end{align}
thus, $h(\lambda) \leq h(\rho)=0$. Then, $f(\lambda,1,\rho) \geq f(\lambda,\frac{1-\rho}{1-\lambda},\lambda)$. The argument goes through also for the case when $\rho$ is negative, which completes the proof.

\section*{Acknowledgment}
This work was supported in part by the Swiss National Science Foundation under Grant 169294, Grant P2ELP2\_165137. 
\bibliographystyle{IEEEtran}
\bibliography{nit_wyn,caching_giel}

\end{document}

%% file: defns.tex
\let\oldbrace\{
\def\{{\oldbrace\kern0.5pt}



%% file: gray_wyner.tex
\begin{tikzpicture}
\draw[->, line width=1pt](0,0)--(9,0);
\draw[->, line width=1pt](0,0)--(0,6.5);
\draw[black,thick,anchor=north] (8.5,0) node{$\Delta e^{\alpha}$};
\draw[black,thick,anchor=center] (0,7) node{$\mathrm{R}_{\Delta ,\alpha}(X,Y)$};
\pgfplotsset{ticks=none}
\begin{axis}[%
hide y axis,
hide x axis,
width=4in,
height=3in,
xmin=0,
xmax=1,
ymin=0,
ymax=4.5,
xlabel={$\gamma_i$},
ylabel={$C_{\gamma_i}({\ve X}_i,{\ve Y}_i)$},
legend style={at={(0.8,0.8)}, legend cell align=center, align=center, draw=white!15!black}
]
\addplot [line width=1pt, draw=none, mark=asterisk, mark options={solid, black}]
  table[row sep=crcr]{%
0.5	0.549306144334055\\
0.51	0.529695787757414\\
0.52	0.510825623765991\\
0.53	0.492641801680553\\
0.54	0.475096141774918\\
0.55	0.458145365937077\\
0.56	0.441750454525582\\
0.57	0.425876105368292\\
0.58	0.410490276034915\\
0.59	0.395563794460075\\
0.6	0.381070026023448\\
0.61	0.3669845875401\\
0.62	0.353285100446043\\
0.63	0.339950976904962\\
0.64	0.326963233703332\\
0.65	0.314304329711187\\
0.66	0.301958023416001\\
0.67	0.289909247626471\\
0.68	0.278143998921374\\
0.69	0.266649239809025\\
0.7	0.255412811882995\\
0.71	0.244423358523608\\
0.72	0.233670255913126\\
0.73	0.22314355131421\\
0.74	0.212833907712842\\
0.75	0.202732554054082\\
0.76	0.192831240405992\\
0.77	0.183122197477442\\
0.78	0.173598099992094\\
0.79	0.164252033486018\\
0.8	0.15507746415192\\
0.81	0.146068211400581\\
0.82	0.13721842285088\\
0.83	0.128522551494945\\
0.84	0.119975334815296\\
0.85	0.111571775657105\\
0.86	0.1033071246815\\
0.87	0.0951768642456094\\
0.88	0.087176693572389\\
0.89	0.0793025150883192\\
0.9	0.0715504218203367\\
0.91	0.0639166857549423\\
0.92	0.0563977470726723\\
0.93	0.0489902041801019\\
0.94	0.0416908044695256\\
0.95	0.0344964357434758\\
0.96	0.0274041182474975\\
0.97	0.0204109972601276\\
0.98	0.0135143361939597\\
0.99	0.00671151016607036\\
1	0\\
};
\addlegendentry{$(1-\rho)\sigma^2\leq \Delta e^{\alpha}\leq \sigma^2$}

\addplot [line width=1pt, color=black, mark=o, mark options={solid, black}]
  table[row sep=crcr]{%
0	inf\\
0.01	4.4613291497622\\
0.02	3.76818196920226\\
0.03	3.36271686109409\\
0.04	3.07503478864231\\
0.05	2.8518912373281\\
0.06	2.66956968053415\\
0.07	2.51541900070689\\
0.08	2.38188760808237\\
0.09	2.26410457242598\\
0.1	2.15874405676816\\
0.11	2.06343387696383\\
0.12	1.9764224999742\\
0.13	1.89637979230066\\
0.14	1.82227182014694\\
0.15	1.75327894865999\\
0.16	1.68874042752242\\
0.17	1.62811580570598\\
0.18	1.57095739186604\\
0.19	1.51689017059576\\
0.2	1.46559687620821\\
0.21	1.41680671203878\\
0.22	1.37028669640389\\
0.23	1.32583493383305\\
0.24	1.28327531941426\\
0.25	1.242453324894\\
0.26	1.20323261174072\\
0.27	1.16549228375787\\
0.28	1.129124639587\\
0.29	1.09403331977573\\
0.3	1.06013176810005\\
0.31	1.02734194527705\\
0.32	0.995593246962474\\
0.33	0.964821588295721\\
0.34	0.93496862514604\\
0.35	0.905981088272787\\
0.36	0.877810211306091\\
0.37	0.850411237117977\\
0.38	0.823742990035815\\
0.39	0.797767503632554\\
0.4	0.772449695648264\\
0.41	0.747757083057893\\
0.42	0.723659531478833\\
0.43	0.700129034068639\\
0.44	0.67713951584394\\
0.45	0.654666659991881\\
0.46	0.632687753273106\\
0.47	0.611181548052142\\
0.48	0.59012813885431\\
0.49	0.569508851651574\\
0.5	0.549306144334055\\
};
\addlegendentry{$0\leq \Delta e^{\alpha} \leq (1-\rho)\sigma^2$}
\end{axis}
\end{tikzpicture}%

%% file: new_CISS_arXiv.bbl
\begin{thebibliography}{1}
\providecommand{\url}[1]{#1}
\csname url@samestyle\endcsname
\providecommand{\newblock}{\relax}
\providecommand{\bibinfo}[2]{#2}
\providecommand{\BIBentrySTDinterwordspacing}{\spaceskip=0pt\relax}
\providecommand{\BIBentryALTinterwordstretchfactor}{4}
\providecommand{\BIBentryALTinterwordspacing}{\spaceskip=\fontdimen2\font plus
\BIBentryALTinterwordstretchfactor\fontdimen3\font minus
  \fontdimen4\font\relax}
\providecommand{\BIBforeignlanguage}[2]{{%
\expandafter\ifx\csname l@#1\endcsname\relax
\typeout{** WARNING: IEEEtran.bst: No hyphenation pattern has been}%
\typeout{** loaded for the language `#1'. Using the pattern for}%
\typeout{** the default language instead.}%
\else
\language=\csname l@#1\endcsname
\fi
#2}}
\providecommand{\BIBdecl}{\relax}
\BIBdecl

\bibitem{Slepian--Wolf}
D.~Slepian and J.~Wolf, ``Noiseless coding of correlated information sources,''
  \emph{IEEE Transactions on Information Theory}, vol.~19, no.~4, pp. 471--480,
  1973.

\bibitem{Gray--Wyner}
R.~M. Gray and A.~D. Wyner, ``Source coding for a simple network,'' \emph{The
  Bell System Technical Journal}, vol.~53, no.~9, pp. 1681 -- 1721, 1974.

\bibitem{Wang--Lim--Gastpar}
C.-Y. Wang, S.~H. Lim, and M.~Gastpar, ``Information-theoretic caching:
  Sequential coding for computing,'' \emph{IEEE Transactions on Information
  Theory}, vol.~62, no.~11, pp. 6393 -- 6406, August 2016.

\bibitem{EPFL8001}
G.~J. {Op 't Veld}, ``\BIBforeignlanguage{eng}{A {G}aussian source coding
  perspective on caching and total correlation},'' Ph.D. dissertation, EPFL
  (Lausanne), 2017.

\bibitem{Hyper_Gauss}
\BIBentryALTinterwordspacing
C.~Nair. An extremal inequlity related to hypercontractivity of {G}aussian
  random variables. [Online]. Available:
  \url{http://chandra.ie.cuhk.edu.hk/pub/papers/manuscripts/ITA14.pdf}
\BIBentrySTDinterwordspacing

\bibitem{Thomas}
J.~Thomas, ``Feedback can at most double gaussian multiple access channel
  capacity (corresp.),'' \emph{IEEE Transactions on Information Theory},
  vol.~33, no.~5, pp. 711 -- 716, September 1987.

\bibitem{Geng--Nair}
Y.~Geng and C.~Nair, ``The capacity region of the two-receiver gaussian vector
  broadcast channel with private and common messages,'' \emph{IWCIT}, vol.~60,
  no.~4, April 2014.

\bibitem{Posner}
E.~Posner, ``Random coding strategies for minimum entropy,'' \emph{IEEE
  Transactions on Information Theory}, vol.~21, no.~4, pp. 388 -- 391, 1975.

\end{thebibliography}
